\documentclass[11pt]{article}
\usepackage{hyperref}
\usepackage{times}  
\usepackage{mathpazo}
\usepackage{amssymb,amsmath,amsthm}
\usepackage{epsfig}

\newtheorem{theorem}{Theorem}[section]

\newtheorem{definition}[theorem]{Definition}

\newtheorem{claim}[theorem]{Claim}
\newtheorem{lemma}[theorem]{Lemma}
\newtheorem{conjecture}[theorem]{Conjecture}

\newtheorem{corollary}[theorem]{Corollary}

\newtheorem{remark}[theorem]{Remark}

\newcommand{\qedsymb}{\hfill{\rule{2mm}{2mm}}}
\renewenvironment{proof}[1][]{\begin{trivlist}
\item[\hspace{\labelsep}{\bf\noindent Proof#1:\/}] }{\qedsymb\end{trivlist}}

\def\calG{{\cal G}}
\def\calF{{\cal F}}

\def\calI{{\cal I}}
\def\calP{{\cal P}}

\def\Z{{\mathbb{Z}}}

\newcommand{\eps}{\epsilon}
\renewcommand{\epsilon}{\varepsilon}

\newcommand{\dist}{\mathop{\mathrm{dist}}}

\newcommand{\Fset}{\mathbb{F}}         

\begin{document}

\title{{\bf Sunflowers and Testing Triangle-Freeness of Functions}}

\author{
Ishay Haviv\thanks{School of Computer Science, The Academic College of Tel Aviv-Yaffo, Tel Aviv 61083, Israel.}
\and Ning Xie\thanks{SCIS, Florida International University, Miami, FL 33199, USA. Research supported in part by NSF grant 1423034. Email: {\tt nxie@cis.fiu.edu}}
}

\date{}

\maketitle

\begin{abstract}
A function $f: \Fset_2^n \rightarrow \{0,1\}$ is {\em triangle-free} if there are no $x_1,x_2,x_3 \in \Fset_2^n$ satisfying $x_1+x_2+x_3=0$ and $f(x_1)=f(x_2)=f(x_3)=1$. In testing triangle-freeness, the goal is to distinguish with high probability triangle-free functions from those that are $\eps$-far from being triangle-free. It was shown by Green that the query complexity of the canonical tester for the problem is upper bounded by a function that depends only on $\eps$ (GAFA,~2005), however the best known upper bound is a tower type function of $1/\eps$. The best known lower bound on the query complexity of the canonical tester is $1/\eps^{13.239}$ (Fu and Kleinberg,~RANDOM,~2014).

In this work we introduce a new approach to proving lower bounds on the query complexity of triangle-freeness. We relate the problem to combinatorial questions on collections of vectors in $\Z_D^n$ and to {\em sunflower conjectures} studied by Alon, Shpilka, and Umans (Comput. Complex.,~2013). The relations yield that a refutation of the Weak Sunflower Conjecture over $\Z_4$ implies a super-polynomial lower bound on the query complexity of the canonical tester for triangle-freeness. Our results are extended to testing $k$-cycle-freeness of functions with domain $\Fset_p^n$ for every $k \geq 3$ and a prime $p$. In addition, we generalize the lower bound of Fu and Kleinberg to $k$-cycle-freeness for $k \geq 4$ by generalizing the construction of uniquely solvable puzzles due to Coppersmith and Winograd (J. Symbolic Comput.,~1990).
\end{abstract}

\section{Introduction}

The research on {\em property testing}, initiated by Rubinfeld and Sudan~\cite{RubinfeldS96} and by Goldreich, Goldwasser, and Ron~\cite{GoldreichGR98}, is concerned with very efficient algorithms that distinguish with high probability objects which satisfy a given property from those that are far from satisfying it. Typically, one can think of an input object as a function from a domain $D$ to a range $R$, and of a property $\calP$ as a subset of the function set $D \rightarrow R$. For a distance parameter $\eps$, the goal of the randomized algorithm, called a {\em tester}, is to accept the functions of $\calP$ and to reject functions which are $\eps$-far from $\calP$, that is, disagree with any function in $\calP$ on at least $\eps$-fraction of the inputs. In case that the functions of $\calP$ are always accepted, we say that the tester has {\em one-sided} error. The main objective in property testing is to minimize the number of queries that the tester makes to the input object. If the number of queries depends solely on the distance parameter $\eps$, we say that the property is {\em strongly testable}.

Since the invention of the property testing model, many natural properties were shown to be strongly testable. A considerable amount of attention was given to testing {\em graph} properties, and the strongly testable dense graph properties were fully characterized~\cite{AlonFNS09,BorgsCLSSV06}. An important graph property testing problem is that of deciding if a given undirected graph is $H$-free, i.e., contains no subgraph isomorphic to $H$, or is $\eps$-far from $H$-freeness, where $H$ is a fixed graph. Whereas $H$-freeness is known to be strongly testable for every graph $H$, it turns out that the graph $H$ significantly affects the dependence of the query complexity on $\eps$. Alon proved in~\cite{Alon02} that for every bipartite graph $H$, the one-sided error query complexity of testing $H$-freeness is polynomial in $1/\eps$, and that for every non-bipartite graph $H$, it is super-polynomial in $1/\eps$, namely, at least $(1/\eps)^{\Omega(\log(1/\eps))}$. Interestingly, the lower bound for the non-bipartite case relies on a construction of Behrend~\cite{Behrend46} of dense sets of integers with no $3$-term arithmetic progressions (see also~\cite{SalemS42,Elkin11}) and on an extension of this construction given in~\cite{Alon02}.

Kaufman and Sudan initiated in~\cite{KaufmanS08} a systematic study of testing {\em algebraic} properties of functions with domain $\Fset^n$ for a field $\Fset$. They considered the class of {\em linear invariant} properties, those which are closed under all linear transformations of the domain, and asked for necessary and sufficient conditions for their strong testability (see~\cite{Sudan11,Bhattacharyya13} for relevant surveys). This class includes the properties that can be described as freeness of solutions to (possibly infinite) systems of linear equations, which were shown to be strongly testable in~\cite{Shapira10} and in~\cite{KSV12} (see also~\cite{BhattacharyyaGS10}). As opposed to the $H$-freeness property of graphs, it is not known which of these properties have query complexity polynomial in $1/\eps$. The $k$-cycle-freeness properties, whose query complexity is the focus of the current work, fall into this category and are described next.

\subsection{Testing $k$-Cycle-Freeness of Boolean Functions}

Let $n$ and $k \geq 3$ be integers, and let $\Fset_p$ be the finite field of prime order $p$. A {\em $k$-cycle} of a function $f: \Fset_p^n \rightarrow \{0,1\}$ is defined as $k$ vectors $x_1,\ldots,x_k \in \Fset_p^n$ satisfying $x_1+\cdots+x_k = 0$ and $f(x_i)=1$ for every $1 \leq i \leq k$. In case that $f$ has no $k$-cycles, we say that it is {\em $k$-cycle-free}. In the property testing problem of $k$-cycle-freeness over $\Fset_p$, the input is a function $f: \Fset_p^n \rightarrow \{0,1\}$, and the goal is to distinguish with high probability $k$-cycle-free functions from those that are $\eps$-far from every $k$-cycle-free function.

In the {\em multiple-function variant} of the $k$-cycle-freeness problem, the input is a $k$-tuple of functions $f_1,\ldots,f_k: \Fset_p^n \rightarrow \{0,1\}$, and a $k$-cycle is defined as $k$ vectors $x_1,\ldots,x_k \in \Fset_p^n$ satisfying $x_1+\cdots+x_k = 0$ and $f_i(x_i)=1$ for every $1 \leq i \leq k$. The goal here is to distinguish $k$-cycle-free $k$-tuples of functions from those that are $\eps$-far from $k$-cycle-freeness, that is, at least $\eps \cdot p^n$ values returned by the functions $f_1,\ldots,f_k$ should be changed in order to make the $k$-tuple free of $k$-cycles. Clearly, the query complexity of the multiple-function variant of $k$-cycle-freeness is at least as large as that of the one-function variant. We observe, though, that whenever $p$ does not divide $k$, the two variants of testing $k$-cycle-freeness are essentially equivalent. On the other hand, in case that $p$ does divide $k$, the one-function variant of the problem seems to be easier, and has query complexity $O(1/\eps)$ (see Section~\ref{sec:multiple} for details). Therefore, in order to understand the query complexity of testing $k$-cycle-freeness in the one-function case, it suffices to understand it for the multiple-function case. The latter is more convenient to deal with while studying lower bounds, so from now on, unless otherwise specified, we refer to the multiple-function variant as the $k$-cycle-freeness problem.

A natural tester for $k$-cycle-freeness over $\Fset_p$, known as the {\em canonical tester} of the problem, repeatedly picks independently and uniformly at random $k-1$ vectors $x_1,\ldots,x_{k-1} \in \Fset_p^n$ and checks if they form, together with $-x_1-\cdots-x_{k-1}$, a $k$-cycle of the functions $f_1,\ldots,f_k$. If no $k$-cycle is found the tester accepts and otherwise it rejects. Despite the simplicity of this one-sided error tester, Green proved in~\cite{Green05} that for every $k$ it has a constant probability of success for query complexity that depends only on $\eps$, hence the $k$-cycle-freeness property is strongly testable. However, the query complexity achieved by Green has a huge dependence on $\eps$, namely, it is a tower of twos whose height is polynomial in $1/\eps$. An improved upper bound on the tower's height follows from results of~\cite{Fox11} and~\cite{KralSV09} (see~\cite[Section~5]{Fox11} and~\cite{HST13}).

The study of lower bounds on the query complexity of testing $k$-cycle-freeness was initiated by Bhattacharyya and Xie in~\cite{BX10}, where the case of triangles over $\Fset_2$ was considered. They provided the first non-trivial lower bound of $1/\eps^{4.847}$ on the query complexity of the canonical tester for triangle-freeness over $\Fset_2$. They also studied connections between the query complexity of the canonical tester for $k$-cycle-freeness over $\Fset_2$ to that of more general testers for the problem.

The proof technique of the above lower bound involved the notion, introduced in~\cite{BX10}, of {\em perfect-matching-free} families of vectors (PMFs). Roughly speaking, a PMF (for triangles over $\Fset_2$) is a collection $\{(a_i,b_i,c_i)\}_{i \in [m]}$ of $m$ triples of vectors in $\Fset_2^n$ satisfying that $a_{i_1}+b_{i_2}+c_{i_3} = 0$ if and only if $i_1 = i_2 = i_3$. This means that the functions $f_1,f_2,f_3 : \Fset_2^n \rightarrow \{0,1\}$ defined as the characteristic functions of the $a_i$'s, $b_i$'s, and $c_i$'s respectively, have $m$ triangles which are pairwise disjoint.
The distance of these functions from triangle-freeness is relatively large compared to the number of their triangles. Hence, they can be used to obtain lower bounds on the query complexity of the canonical tester for triangle-freeness over $\Fset_2$. The authors of~\cite{BX10} used a computer search to construct a PMF of vectors in $\Fset_2^n$ of size (roughly) $1.67^n$, and this allowed them to get their $1/\eps^{4.847}$ lower bound. Further, they showed that a PMF of size $(2-o(1))^{n}$ implies a super-polynomial lower bound on the query complexity of the canonical tester for triangle-freeness over $\Fset_2$, and conjectured that such a PMF exists.

Very recently, Fu and Kleinberg~\cite{FuK13} discovered an interesting connection between PMFs and the combinatorial objects known as {\em uniquely solvable puzzles} (USPs). The latter were introduced in the context of matrix multiplication algorithms and were explicitly defined by Cohn et al. in~\cite{CohnKSU05}. Coppersmith and Winograd~\cite{CoppersmithW90} implicitly gave a probabilistic construction of $n$-dimensional USPs of size $(3/2^{2/3}-o(1))^{n} \approx 1.89^n$ that played a central role in their famous $O(n^{2.376})$-time algorithm for multiplication of $n$ by $n$ matrices, whose running time was improved only two decades later~\cite{stothers,Williams12}. It was shown in~\cite{FuK13} that every USP implies a PMF of the same cardinality, and this led to an improved lower bound of $1/\eps^{13.239}$ on the query complexity of the canonical tester for triangle-freeness over $\Fset_2$. However, it was observed in~\cite{CohnKSU05} that the USP construction of~\cite{CoppersmithW90} is essentially optimal, hence it seems that the USP-based approach to proving lower bounds on testing triangle-freeness has been pushed to its limit, and, in particular, cannot yield super-polynomial lower bounds. Yet, a strengthened notion of USP, known as {\em strong} USP, was studied by Cohn et al.~\cite{CohnKSU05}, who proved that if strong USPs of optimal size exist then the exponent of matrix multiplication is $2$. A fascinating challenge, which was left open in~\cite{FuK13}, is to show that strong USPs might imply super-polynomial lower bounds on the query complexity of related testing problems.

In the current work we show that lower bounds on testing $k$-cycle-freeness might follow from the existence of certain collections of vectors in $\Z_D^n$. These collections are related to famous sunflower conjectures, which we turn to describe in the next section.

\subsection{Sunflower Conjectures}\label{sec:sunflowerIntro}

A {\em $k$-sunflower} is a collection of $k$ sets that have the same pairwise intersections. This notion was introduced in 1960 by Erd{\"o}s and Rado~\cite{ErdosR60}, and besides being of great interest in combinatorics, it found applications in several areas of computer science, e.g., circuit complexity~\cite{Razborov85,AlonB87}, hardness of approximation~\cite{DinurSafra05}, and property testing~\cite{AlonHW13}. The main question regarding $k$-sunflowers is how large a collection of sets containing no $k$-sunflowers can be. It was shown in~\cite{ErdosR60} that the size of any collection of subsets of size $s$ of some universe $U$ with no $k$-sunflowers is at most $s! \cdot (k-1)^s$. The classical sunflower conjecture of Erd{\"o}s and Rado asserts the following.

\begin{conjecture}[Classical Sunflower Conjecture~\cite{ErdosR60}]\label{conj:classical}
For every $k>0$ there exists a constant $c_k$, such that every collection of at least $c_k^s$ subsets of size $s$ of some universe $U$ contains a $k$-sunflower.
\end{conjecture}

The above conjecture is still open even for the special case of $k=3$. For this case, Kostochka showed an improved upper bound of $c \cdot s! \cdot \Big(\frac{30\ln \ln \ln s}{\ln \ln s}\Big)^s$ for some constant $c>0$~\cite{Kostochka97}. Erd{\"o}s and Szemer{\'e}di~\cite{ErdosS78} presented in 1978 the following conjecture on $3$-sunflowers inside $[n]$, and proved that Conjecture~\ref{conj:classical}, even restricted to $k=3$, implies it.
\begin{conjecture}[Sunflower Conjecture in $\{0,1\}^n$~\cite{ErdosS78}]\label{conj:0,1}
There exists an $\eps >0$, such that every collection $\calF$ of subsets of $[n]$ ($n \geq 2$) of size $|\calF| \geq 2^{(1-\eps)n}$ contains a $3$-sunflower.
\end{conjecture}

In a recent paper, Alon, Shpilka, and Umans~\cite{AlonSU13} have studied a new notion of sunflowers over $\Z_D = \{1,\ldots,D\}$ and several related sunflower conjectures. Following their definition, we say that $k$ vectors $v_1,\ldots,v_k$ in $\Z_D^n$ form a {\em $k$-sunflower} if for every $i \in [n]$ it holds that the $i$th entries $(v_1)_i,\ldots,(v_k)_i$ of the vectors are either all equal or all distinct. It was shown in~\cite{AlonSU13} that Conjecture~\ref{conj:0,1} can be {\em equivalently} formulated in terms of sunflowers of vectors as follows.
\begin{conjecture}[Sunflower Conjecture in $\Z_D^n$~\cite{AlonSU13}]\label{conj:sunflowerD}
There exist $\eps > 0$, $D_0$ and $n_0$, such that for every $D \geq D_0$ and $n \geq n_0$, every collection $\calF$ of vectors in $\Z_D^n$ of size $|\calF| \geq D^{(1-\eps)n}$ contains a $3$-sunflower.
\end{conjecture}

The above conjecture, just like Conjecture~\ref{conj:0,1}, is widely believed to be true. Still, one might wonder if its assertion holds for small values of $D$. It is stated below for a specific integer $D$.
\begin{conjecture}[Weak Sunflower Conjecture over $\Z_D$~\cite{AlonSU13}]\label{conj:weakD}
There exist $\eps > 0$ and $n_0$, such that for every $n \geq n_0$, every collection $\calF$ of vectors in $\Z_D^n$ of size $|\calF| \geq D^{(1-\eps)n}$ contains a $3$-sunflower.
\end{conjecture}
\noindent
Of special importance is the Weak Sunflower Conjecture over $\Z_3$, which refers to the maximum possible size of a collection of vectors in the group $\Z_3^n$ with no $3$-term arithmetic progressions (or, equivalently, non-trivial triples of vectors with zero sum modulo $3$). The largest known construction of such collections has cardinality $c^n$ for $c \approx 2.217$~\cite{Edel04}. An upper bound of $2 \cdot 3^n /n$ was shown by Meshulam in~\cite{Meshulam95} (see~\cite{LiuS09} for a generalization of his result), and this was recently improved by Bateman and Katz to $O(3^n / n^{1+\eps})$ for some constant $\eps >0$~\cite{BatemanK12}.

Given the similarity between Conjecture~\ref{conj:sunflowerD} and the Weak Sunflower Conjecture over $\Z_D$, one might guess that the latter is true for small values of $D$. In fact, it was observed in~\cite{AlonSU13} that for {\em every} $D \geq 3$, the assertion of the Weak Sunflower Conjecture over $\Z_D$ implies Conjecture~\ref{conj:sunflowerD}. Nevertheless, Conjecture~\ref{conj:sunflowerD} seems to be much more likely to hold than Conjecture~\ref{conj:weakD} for small values of $D$. For example, as explained in~\cite{AlonSU13}, the case of $D=3$ can be viewed as a variant of the assertion that collections of $D^{(1-\eps)n}$ vectors in $\Z_D^n$ must contain a $3$-term arithmetic progression modulo $D$. However, a result of Salem and Spencer~\cite{SalemS42} implies that the latter is false for large values of $D$, namely, for $D > 2^{2/\eps}$.

\subsection{Our Contribution}

In this work we introduce a new approach to proving lower bounds on the query complexity of testing $k$-cycle-freeness over $\Fset_p$ for general $k \geq 3$ and primes $p$. To do so, we show that certain collections of vectors in $\Z_D^n$, which are related to some of the sunflower conjectures described above, can be used to obtain perfect-matching-free vector families. For example, for the special case of triangle-freeness over $\Fset_2$, it is shown that a large collection of vectors in $\Z_4^n$ containing no $3$-sunflowers implies a large perfect-matching-free family over $\Fset_2$, thus implying a lower bound on testing triangle-freeness. In case that the size of the collection is $(4-o(1))^{n}$, it yields a super-polynomial lower bound, as stated below.

\begin{theorem}\label{thm:F2Intro}
If the Weak Sunflower Conjecture (Conjecture~\ref{conj:weakD}) over $\Z_{4}$ is false, then the query complexity of the canonical tester for triangle-freeness over $\Fset_2$ for distance $\eps$ is super-polynomial in $1/\eps$.
\end{theorem}
\noindent
As alluded to before, a refutation of the Weak Sunflower Conjecture over $\Z_D$ for small values of $D$ would not be overly surprising, thus a super-polynomial lower bound on testing triangle-freeness over $\Fset_2$ might stem from the above theorem. Yet, even if the Weak Sunflower Conjecture over $\Z_4$ is true, large collections of vectors in $\Z_4^n$ containing no $3$-sunflowers might provide improvements on the known lower bounds. Specifically, our results imply that any such collection of size $(c-o(1))^{n}$ for $c> 9/2^{4/3} \approx 3.57$ beats the best known lower bound of~\cite{FuK13}, but for $c < 4$ the obtained lower bound is only polynomial in $1/\eps$ (see Theorem~\ref{thm:Mainp^2}).

We then generalize Theorem~\ref{thm:F2Intro} in a couple of ways. First, we obtain the following extension to triangle-freeness over $\Fset_p$, where $p$ is an arbitrary prime.
\begin{theorem}\label{thm:FpIntro}
For every prime $p$, if the Weak Sunflower Conjecture (Conjecture~\ref{conj:weakD}) over $\Z_{p^2}$ is false, then the query complexity of the canonical tester for triangle-freeness over $\Fset_p$ for distance $\eps$ is super-polynomial in $1/\eps$.
\end{theorem}

Theorem~\ref{thm:FpIntro} implies that for every prime $p$, a refutation of a certain Weak Sunflower Conjecture over $\Z_D$ (for $D=p^2$) implies a super-polynomial lower bound on the query complexity of the canonical tester for triangle-freeness over $\Fset_p$. Therefore, the unlikely event that for {\em some} prime $p$ the query complexity is polynomial, implies Conjecture~\ref{conj:sunflowerD}. On the other hand, it was shown in~\cite{AlonSU13} that a conjecture of Coppersmith and Winograd, which was shown in~\cite{CoppersmithW90} to imply that the matrix multiplication exponent is $2$, implies that Conjecture~\ref{conj:sunflowerD} is false. Hence, the conjecture of~\cite{CoppersmithW90} implies, if true, a super-polynomial lower bound on the number of queries made by the canonical tester for testing triangle-freeness over $\Fset_p$ for every prime $p$.

We note that for $p=3$ one can show a stronger statement than that of Theorem~\ref{thm:FpIntro}. Indeed, in this case a super-polynomial lower bound follows quite easily from a refutation of the Weak Sunflower Conjecture over $\Z_3$ (which can be only weaker than its refutation over $\Z_9$; see Section~\ref{sec:sunfPMF}). Interestingly, Alon et al.~\cite{AlonSU13} studied a variant of this conjecture, called the {\em Multicolored} Sunflower Conjecture over $\Z_3$, and related it to the notion of {\em strong} uniquely solvable puzzles. It turns out that this multicolored conjecture coincides with our question on perfect-matching-free families over $\Fset_3$, and that their results imply an (unconditional) lower bound of $1/\eps^{7.298}$ on the query complexity of the canonical tester for triangle-freeness over $\Fset_3$. Moreover, we use a result of~\cite{AlonSU13} and the connection observed here to obtain that if the conjecture of~\cite{CohnKSU05} that strong USPs of optimal size exist is true, then the query complexity of the canonical tester for triangle-freeness over $\Fset_3$ is super-polynomial. This gives, in a sense, an affirmative answer to a question posed in~\cite{FuK13}.

The above results are also extended to testing $k$-cycle-freeness for every $k \geq 3$. We show how lower bounds on the query complexity of the canonical tester for $k$-cycle-freeness over $\Fset_p$ follow from the existence of certain collections of vectors in $\Z_D^n$ for an appropriate choice of $D$. Namely, we are interested in collections of vectors in $\Z_D^n$ for $D = p^{k-1}$, satisfying that for every $k$ vectors in the collection (not all equal) there is some $i \in [n]$ for which the $k$ vectors have exactly {\em two} distinct symbols in their $i$th entries. Notice that for $k=3$ this simply means that the collection contains no $3$-sunflowers. As before, for vector collections of optimal size $(D-o(1))^n$, the obtained lower bound on the query complexity turns out to be super-polynomial (see Section~\ref{sec:sunfPMF}).

Finally, we show that the lower bound of Fu and Kleinberg~\cite{FuK13} on testing triangle-freeness over $\Fset_2$ can be generalized to testing $k$-cycle-freeness over $\Fset_p$.
\begin{theorem}\label{thm:GenkIntro}
For every $k \geq 3$ and a prime $p$, the query complexity of the canonical tester for $k$-cycle-freeness over $\Fset_p$ for distance $\eps$ is $\Omega(1/\eps^{g(k)-o(1)})$ for
$$g(k) = \frac{k-1-H(1/k) / \log_2{p}}{1-H(1/k) / \log_2{p}},$$
where $H$ stands for the binary entropy function.
\end{theorem}
\noindent
The proof of Theorem~\ref{thm:GenkIntro} relies on a delicate extension of the construction of Coppersmith and Winograd~\cite{CoppersmithW90} of uniquely solvable puzzles, which is based on a construction of Behrend~\cite{Behrend46}, which found great interest in additive combinatorics. Interestingly, our construction requires an extension of Behrend's result, given in~\cite{Alon02}, that was used there for proving lower bounds on testing $H$-freeness of graphs.

\subsection{Outline}
The rest of the paper is organized as follows. In Section~\ref{sec:preli} we provide a background on the $k$-cycle-freeness problem, relate its one-function and multiple-function variants, present the notion of perfect-matching-free families of vectors (PMFs), and show how they imply lower bounds on the query complexity of the problem. In Section~\ref{sec:sunfPMF} we prove that PMFs can be constructed using certain collections of vectors in $\Z^n_D$ and derive relations to sunflower conjectures, including Theorem~\ref{thm:FpIntro}. Finally, in Section~\ref{sec:Genk}, we prove Theorem~\ref{thm:GenkIntro}.

\section{Testing $k$-Cycle-Freeness of Boolean Functions}\label{sec:preli}

Let $n$ and $k \geq 3$ be integers, and let $\Fset_p$ be the finite field of prime order $p$. A {\em $k$-cycle} of $k$ functions $f_1,\ldots,f_k: \Fset_p^n \rightarrow \{0,1\}$ is defined as $k$ vectors $x_1,\ldots,x_k \in \Fset_p^n$ satisfying $x_1+\cdots+x_k = 0$ and $f_i(x_i)=1$ for every $1 \leq i \leq k$. If a $k$-tuple of functions $(f_1,\ldots,f_k)$ has no $k$-cycles, we say that it is {\em $k$-cycle-free}. Its {\em distance} from $k$-cycle-freeness is defined as $$\min_{(g_1,\ldots,g_k)}{\sum_{i=1}^{k}{{\dist}(f_i,g_i)}},$$
where the minimum is over all the $k$-cycle-free $k$-tuples of functions $(g_1,\ldots,g_k)$, and $\dist(f,g)$ denotes the fraction of points at which the functions $f$ and $g$ disagree. We say that a $k$-tuple of functions is {\em $\eps$-far} from $k$-cycle-freeness if its distance from $k$-cycle-freeness is at least $\eps$.

In the property testing problem of $k$-cycle-freeness over $\Fset_p$, the input is a $k$-tuple of functions $f_1,\ldots,f_k: \Fset_p^n \rightarrow \{0,1\}$, and the goal is to accept $k$-cycle-free $k$-tuples of functions with probability at least $2/3$ and to reject $k$-tuple of functions which are $\eps$-far from $k$-cycle-freeness with probability at least $2/3$. The {\em canonical tester} for $k$-cycle-freeness over $\Fset_p$ repeatedly picks uniformly and independently $k$ vectors with zero sum and checks if they form a $k$-cycle of the input functions. If no $k$-cycle is found the tester accepts and otherwise it rejects.

\subsection{Multiple-function vs. One-function}\label{sec:multiple}

As mentioned before, one might consider the one-function variant of the $k$-cycle-freeness testing problem. A {\em $k$-cycle} of a function $f:\Fset_p^n \rightarrow \{0,1\}$ is defined as $k$ vectors that sum to the zero vector and are all mapped by $f$ to $1$. The input of the one-function variant is a single function $f:\Fset_p^n \rightarrow \{0,1\}$, and the goal is to decide if $f$ is $k$-cycle-free or $\eps$-far from every $k$-cycle-free function. The canonical tester for $k$-cycle-freeness is naturally extended to the one-function case.

We observe that whenever $p$ does not divide $k$, every $k$-tuple of functions can be transformed to a single function with the same number of $k$-cycles, a similar domain size, and a similar distance from $k$-cycle freeness. This implies that, in this case, the canonical testers for the multiple-function and the one-function variants of the problem have essentially the same query complexity.

\begin{lemma}\label{lemma:p_not_div_k}
Let $n$ be a positive integer, let $k \geq 3$ and $p$ be fixed integers such that $p$ is a prime that does not divide $k$, and let $\alpha >0$ be a real number.
Suppose that the $k$-tuple of functions $f_1, \ldots, f_k : \Fset_p^n \to \{0,1\}$ is $\eps_1$-far from $k$-cycle-freeness and that the canonical tester for $k$-cycle-freeness needs to make $q=\Omega(1/\eps_1^\alpha)$ queries to $(f_1, \ldots, f_k)$. Then, there exists a function $f: \Fset_p^{n+k-1} \rightarrow \{0,1\}$, such that $f$ is $\eps_2$-far from $k$-cycle-freeness for $\eps_2 = \eps_1/p^{k-1}$, and the canonical tester needs to make $\Omega(1/\eps_2^\alpha)$ queries to $f$.
\end{lemma}
\begin{proof}
Given the $k$-tuple of functions $(f_1, \ldots, f_k)$, define $f: \Fset_p^{n+k-1} \to \{0,1\}$ as follows. For all $y\in \Fset_p^n$ and $z\in \Fset_p^{k-1}$, let
\[
f(y, z)=
\begin{cases}
f_i(y), &\text{if $z=e_i$ for $1\leq i\leq k-1$,}\\
f_k(y), &\text{if $z=-e_1-\cdots -e_{k-1}$,}\\
0, &\text{otherwise,}
\end{cases}
\]
where $e_i$ denotes the vector whose entries are all $0$ except the $i$th which is $1$.

First, observe that the only way to choose $k$ vectors (repetitions are allowed) from the set
$$\{e_1,\ldots,e_{k-1},-e_1-\cdots-e_{k-1}\},$$
so that their sum is the zero vector over $\Fset_p$, is to choose each of the vectors exactly once (because $p$ does not divide $k$).
This implies that all the $k$-cycles of $f$ have exactly one point in each of the subfunctions $f_1,\ldots,f_k$. Hence there exists a bijection between the $k$-cycles of $f_1,\ldots,f_k$ and those of $f$. Since $(f_1,\ldots,f_k)$ is $\eps_1$-far from $k$-cycle-freeness, it follows that $f$ is $\eps_2$-far from $k$-cycle-freeness for $\eps_2 = \eps_1 / p^{k-1}$.

Let $N_{\text{cycles}}$ be the number of $k$-cycles of $(f_1, \ldots, f_k)$ and of $f$. Since the query complexity of the canonical tester on a $k$-tuple of functions (resp. function) is proportional to the inverse of the number of $k$-cycles of the input $k$-tuple (resp. function), the query complexity on $f$ is $\Omega(q')$ for
\[
q'=p^{(n+k-1)(k-1)}/N_{\text{cycles}}=\Theta(p^{n(k-1)}/N_{\text{cycles}})=\Theta(q)=
\Omega(1/\eps_1^\alpha)=\Omega(1/\eps_2^\alpha).
\]
\end{proof}

In case that the prime $p$ divides $k$, the one-function variant of $k$-cycle-freeness over $\Fset_p$ is quite easy. The reason is that in this case every vector in the support of a function $f:\Fset_p^n \rightarrow \{0,1\}$, taken with multiplicity $k$, forms a $k$-cycle of $f$. Thus, the problem reduces to deciding if the input function is the zero constant function or is $\eps$-far from it. The tester that given a function $f$ picks uniformly and independently $O(1/\eps)$ random vectors in $\Fset_p^n$ and accepts if and only if they are all mapped by $f$ to $0$ implies the following.

\begin{claim}\label{claim:p_div_k_}
Let $k \geq 3$ be an integer divisible by a prime $p$. Then, for every $\eps >0$, there is a one-sided error tester for the one-function variant of $k$-cycle-freeness over $\Fset_p$ for distance $\eps$ with query complexity $O(1/\eps)$.
\end{claim}

One may ask if a similar result can be shown once we consider only {\em non-trivial} cycles of $f$, that is, $k$ vectors, not all equal, that sum to zero and are all mapped by $f$ to $1$. It turns out that if $p$ divides $k$, $O(1/\eps)$ queries are still sufficient to decide if a given function $f:\Fset_p^n \rightarrow \{0,1\}$ is free of non-trivial $k$-cycles or $\eps$-far from this property. The reason is that the density of such functions turns out to be very small, as follows from the following (special case of a) theorem of Liu and Spencer~\cite{LiuS09}.

\begin{theorem}[\cite{LiuS09}]\label{thm:MeshulamGen}
For every $n$ and a prime $p \geq 3$, if $A \subseteq \Fset_p^n$ contains no $p$ vectors, not all equal, whose sum is the zero vector, then $|A| = o(p^n)$.
\end{theorem}

\noindent
Using the above theorem, it can be easily observed that if $p$ divides $k$ and $f:\Fset_p^n \rightarrow \{0,1\}$ is free of non-trivial $k$-cycles, then it is $o(1)$-close to the zero constant function. Thus, by the same tester that was used for Claim~\ref{claim:p_div_k_}, we get query complexity $O(1/\eps)$ and an ``almost'' one-sided error, that is, functions that are free of non-trivial $k$-cycles are accepted with probability that tends to $1$ where $n$ tends to infinity.

\subsection{Perfect-Matching-Free Families}

We now define the notion of {\em local perfect-matching-free} vector families, which can be used to obtain lower bounds on the query complexity of the canonical tester for $k$-cycle-freeness over $\Fset_p$.
\begin{definition}
An $(n,m)$ {\em local perfect-matching-free family (PMF) for $k$-cycles over $\Fset_p$} is a collection $$\{(x^{(1)}_i,x^{(2)}_i,\ldots,x^{(k)}_i)\}_{i \in [m]},$$ such that for every $i \in [m]$, $x^{(1)}_i,x^{(2)}_{i},\ldots,x^{(k)}_i$ are $k$ vectors in $\Fset_p^n$ whose sum is zero, and for every $i_1,i_2,\ldots,i_k \in [m]$, if the sum of the vectors $x^{(1)}_{i_1},x^{(2)}_{i_2},\ldots,x^{(k)}_{i_k}$ is zero then $i_1 = i_2 = \cdots = i_k$.
The {\em local PMF capacity for $k$-cycles over $\Fset_p$} is the largest constant $c$ for which there exist $(n,(c-o(1))^n)$ local PMFs for $k$-cycles over $\Fset_p$ for infinitely many values of $n$.
\end{definition}

Two remarks are in order.

\begin{remark}\label{remark:limit}
If the local PMF capacity for $k$-cycles over $\Fset_p$ is $c$, then there exist $(n,(c-o(1))^n)$ local PMFs for $k$-cycles over $\Fset_p$ for {\em every} sufficiently large value of $n$ (and not only for infinitely many of them). To see this, for every $n$, denote by $m_n$ the largest integer for which there exists an $(n,m_n)$ local PMF for $k$-cycles over $\Fset_p$. Since $m_{n+n'} \geq m_n \cdot m_{n'}$, we may apply Fekete's lemma (see, e.g.,~\cite[Lemma~11.6]{BookFekete}) to show that the limit of $m_n^{1/n}$, as $n$ tends to infinity, exists and equals the capacity $c$.
\end{remark}
\begin{remark}
Our definition of {\em local} PMFs is slightly different from the definition of PMFs given in~\cite{BX10}. Namely, the requirement in the definition of PMFs in~\cite{BX10} is that for every $k$ permutations $\pi_1,\ldots,\pi_k$ of $[m]$, either $\pi_1 = \cdots = \pi_k$, or there exists an $i \in [m]$ for which the sum $x_{\pi_1(i)}^{(1)}+ \cdots+x_{\pi_k(i)}^{(k)}$ is nonzero. Clearly, every local PMF is a PMF. Whereas the other direction does not hold, it is easy to see that the local PMF capacity for $k$-cycles over $\Fset_p$ equals the PMF capacity for $k$-cycles over $\Fset_p$. For completeness, we include a short proof (which resembles that of~\cite[Proposition~6.3]{CohnKSU05}), and throughout the paper we prefer to consider the notion of local PMFs, mainly for simplicity of presentation.
\end{remark}

\begin{claim}
The local PMF capacity for $k$-cycles over $\Fset_p$ equals the PMF capacity for $k$-cycles over $\Fset_p$.
\end{claim}
\begin{proof}
Clearly, the PMF capacity for $k$-cycles over $\Fset_p$ is at least as large as the local PMF capacity for $k$-cycles over $\Fset_p$. For the other direction, let $\calF$ be an $(n,m)$ PMF for $k$-cycles over $\Fset_p$ for $m=(c-o(1))^n$. For every permutation $\pi$ of $[m]$ consider the $k$-tuple of vectors of length $nm$, obtained by concatenating the $m$ $k$-tuples of vectors in $\calF$ ordered according to $\pi$. Let $\calG$ be the collection of all the $k$-tuples obtained this way. Observe that $\calG$ is an $(nm,m!)$ local PMF for $k$-cycles over $\Fset_p$ and that
$$m! = m^{(1-o(1))m} = (c-o(1))^{(1-o(1))nm} = (c-o(1))^{nm}.$$
Thus, the local PMF capacity for $k$-cycles over $\Fset_p$  is at least $c$, and we are done.
\end{proof}

The following lemma and corollary show how local PMFs imply lower bounds for testing $k$-cycle-freeness. Similar statements were shown in~\cite{BX10}, and we include here the proofs for completeness.
\begin{lemma}\label{lemma:LowerBound(m,n)}
Let $k \geq 3$ be an integer, and let $p$ be a prime. Suppose that there exists an $(n,m)$ local PMF for $k$-cycles over $\Fset_p$. Then, the query complexity of the canonical tester for $k$-cycle-freeness over $\Fset_p$ for distance $\eps$ on $n$ variable functions is $\Omega(1/\eps^{\alpha})$ for $\eps = m/p^n$ and $\alpha = \frac{k-1-(\log_p{m})/n}{1-(\log_p{m})/n}$.
\end{lemma}

\begin{proof}
Let $\calF = \{(x^{(1)}_i,x^{(2)}_i,\ldots,x^{(k)}_i)\}_{i \in [m]}$ be an $(n,m)$ local PMF for $k$-cycles over $\Fset_p$. For every $1 \leq j \leq k$, let $f_j:\Fset_p^n \rightarrow \{0,1\}$ be the characteristic function of the set $\{x^{(j)}_i\}_{i \in [m]}$. By definition of local PMFs, the number of $k$-cycles of the $k$-tuple of functions $(f_1,\ldots,f_k)$ is $m$, and these cycles are pairwise disjoint. Hence, in order to remove all the $m$ cycles, one has to change at least $m$ values of the functions, so this $k$-tuple is $\eps$-far from $k$-cycle-freeness for $\eps = \frac{m}{p^n}$. On the other hand, the probability that one iteration of the canonical tester, applied to $(f_1,\ldots,f_k)$, finds a $k$-cycle is $\frac{m}{p^{(k-1)n}}$, so its query complexity is $\Omega(q)$, for
\[q = \frac{p^{(k-1)n}}{m} = p^{(k-1)n - \log_p{m}} = (1/\eps)^{\frac{k-1 - (\log_p{m})/n}{1 - (\log_p{m})/n}}.\]
\end{proof}

\begin{corollary}\label{cor:PMF2LowerBound}
Let $k \geq 3$ and $p$ be fixed integers, such that $p$ is prime. If the local PMF capacity for $k$-cycles over $\Fset_p$ is $c$, then for every $d<c$, the query complexity of the canonical tester for $k$-cycle-freeness over $\Fset_p$ for distance $\eps$ is $\Omega(1/\eps^{\alpha})$ where $\alpha = \frac{k-1-\log_p{d}}{1-\log_p{d}}$. Furthermore, for every sufficiently small $\eps$ there exists an $n_0=n_0(\eps)$ such that for every $n \geq n_0$ the lower bound holds for $k$-tuples of $n$ variable functions that depend on all of their input variables. In particular, if the local PMF capacity for $k$-cycles over $\Fset_p$ is $p$, then the query complexity of the canonical tester for $k$-cycle-freeness over $\Fset_p$ for distance $\eps$ is super-polynomial in $1/\eps$.
\end{corollary}

\begin{proof}
Let $c$ denote the local PMF capacity for $k$-cycles over $\Fset_p$, and take an arbitrary $d<c$. Using Remark~\ref{remark:limit}, for every sufficiently large $n$ there exists an $(n, \lceil d^n \rceil )$ local PMF for $k$-cycles over $\Fset_p$. Now, for a given sufficiently small $\eps$, let $n_0=n_0(\eps)$ be the largest integer satisfying $\eps \leq \frac{ \lceil d^{n_0} \rceil }{p^{n_0}}$. For this $n_0$ there exists an $(n_0,\lceil d^{n_0} \rceil)$ local PMF for $k$-cycle-freeness over $\Fset_p$. By Lemma~\ref{lemma:LowerBound(m,n)}, the corresponding $k$-tuple of functions $f_1,\ldots,f_k : \Fset_p^{n_0} \rightarrow \{0,1\}$ is $\eps$-far from $k$-cycle-freeness and requires $\Omega(1/\eps^{\alpha})$ queries of the canonical tester for $\alpha$ as in the statement of the corollary.

It remains to show that the above lower bound can be extended to $k$-tuples of functions with domain $\Fset_p^n$ for every $n \geq n_0$. For every $1 \leq j \leq k$ define the function $g_j : \Fset_p^n \rightarrow \{0,1\}$ such that $g_j(y)=1$ if and only if $y = (x,z)$ for $x \in \Fset_p^{n_0}$ and $z \in \Fset_p^{n-n_0}$ satisfying $f_j(x)=1$ and $\sum_{i=1}^{n-n_0}{z_i}=0$.\footnote{This is a slight generalization of a construction due to Jakob Nordstr{\"{o}}m (Private communication,~2010).} The $k$-tuple of functions $(g_1,\ldots,g_k)$ has at least $\eps \cdot p^{n_0} \cdot (p^{n-n_0-1})^{k-1}$ $k$-cycles, and every vector of these cycles belongs to $(p^{n-n_0-1})^{k-2}$ of the cycles. Therefore, $(g_1,\ldots,g_k)$ is $\eps'$-far from $k$-cycle-freeness for $\eps' = \eps/p = \Theta(\eps)$ and requires query complexity $\Omega(1/\eps^{\alpha})$, thus the required lower bound holds for every sufficiently small distance parameter. In addition, it is easy to verify that the $k$-tuple $(g_1,\ldots,g_k)$ depends on all of its input variables, as required.

Finally, observe that if the local PMF capacity for $k$-cycles over $\Fset_p$ is $p$, then for every $\alpha>0$, the query complexity of the canonical tester for $k$-cycle-freeness over $\Fset_p$ for some distance $\eps$ is $\Omega(1/\eps^{\alpha})$, thus it is super-polynomial in $1/\eps$.
\end{proof}

We turn to define (strong) uniquely solvable puzzles (USPs). Then we state a theorem of Alon et al.~\cite{AlonSU13} that says that strong USPs imply local PMFs for triangles over $\Fset_3$ (in their language, collections of ordered $3$-sunflowers in $\Z_3^n \times \Z_3^n \times \Z_3^n$ containing no {\em multicolored} sunflowers).

\begin{definition}\label{def:USP}
An $n$-dimensional {\em uniquely solvable puzzle} (USP) is a collection of vectors $\{x_i\}_{i \in [m]}$ in $\Z_3^n$ satisfying that for every three permutations $\pi_1,\pi_2,\pi_3$ of $[m]$, either $\pi_1 = \pi_2 = \pi_3$, or there exist $i \in [m]$ and $j \in [n]$ for which at least two of $(x_{\pi_1(i)})_j = 1$, $(x_{\pi_2(i)})_j = 2$, and $(x_{\pi_3(i)})_j = 3$ hold. A {\em strong USP} is defined similarly replacing the ``at least two'' by ``exactly two''.
The {\em (strong) USP capacity} is the largest constant $c$ for which there exist $n$-dimensional (strong) USPs of size $(c-o(1))^n$ for infinitely many values of $n$.
\end{definition}

\begin{theorem}[\cite{AlonSU13}]\label{thm:PMF3AlonSU}
If the strong USP capacity is at least $c$, then the local PMF capacity for triangles over $\Fset_3$ is at least $2^{2/3} \cdot c$.
\end{theorem}
\noindent
It is known that the strong USP capacity is at least $2^{2/3}$~\cite[Proposition~3.8]{CohnKSU05}. Hence, by Theorem~\ref{thm:PMF3AlonSU}, the local PMF capacity for triangles over $\Fset_3$ is at least $2^{4/3}$. By Corollary~\ref{cor:PMF2LowerBound}, it follows that the query complexity of the canonical tester for triangle-freeness over $\Fset_3$ for distance $\eps$ is at least $1/\eps^{7.298}$. Cohn, Kleinberg, Szegedy, and Umans conjectured that the strong USP capacity is $3/2^{2/3}$ and proved in~\cite{CohnKSU05} that their conjecture implies that the exponent of matrix multiplication is $2$. By Theorem~\ref{thm:PMF3AlonSU}, if their conjecture is true then the local PMF capacity for triangles over $\Fset_3$ is $3$, and the latter yields, by Corollary~\ref{cor:PMF2LowerBound}, a super-polynomial lower bound on the query complexity of the canonical tester for triangle-freeness over $\Fset_3$.

\section{Sunflower Conjectures vs. Local PMFs}\label{sec:sunfPMF}

In this section we prove that local PMFs for $k$-cycles over $\Fset_p$ can be constructed from certain collections of vectors in $\Z_D^n$, some of which are related to sunflower conjectures of Alon et al.~\cite{AlonSU13}. The idea behind the construction is quite simple: every vector of these collections is mapped to a $k$-tuple of vectors, in a way that every symbol of $\Z_D$ is replaced by certain $k$ vectors over $\Fset_p$.

We need the following lemma and the corollary that follows it. We use here the notation $A^{(\ell)}$ to denote the $\ell$th column of a matrix $A$.
\begin{lemma}\label{lemma:constructionA}
For every prime $p$ and a positive integer $k$, there exists a collection of $p^k$ matrices $A_1,A_2,\ldots,A_{p^k}$ in $\Fset_p^{k \times k}$ such that for every $1 \leq i \neq j \leq p^k$ and every non-empty set $I \subseteq [k]$ it holds that $$\sum_{\ell \in I}{A_i^{(\ell)}} \neq \sum_{\ell \in I}{A_j^{(\ell)}}.$$
\end{lemma}
\begin{proof}
Denote $q = p^k$, and let $\alpha_1,\ldots,\alpha_q$ be the $q$ elements of the field $\Fset_q$. Let $enc:\Fset_q \rightarrow \Fset_p^k$ be the natural encoding of the elements of $\Fset_q$ as distinct vectors in $\Fset_p^k$. This encoding is linear, that is, $enc(0)=(0,\ldots,0)$ and $enc(x+y)=enc(x)+enc(y)$ for every $x,y \in \Fset_q$. Let $\beta$ be a generator of the multiplicative group $\Fset_q^*$, and notice that the set $\{1,\beta,\beta^2,\ldots,\beta^{k-1}\}$ is linearly independent over $\Fset_p$.

Now, for every $1 \leq i \leq q$, define the $k$ by $k$ matrix $A_i$ as the matrix whose columns are
$$enc(\alpha_i),enc(\alpha_i \cdot \beta),enc(\alpha_i \cdot \beta^2),\ldots,enc(\alpha_i \cdot \beta^{k-1}).$$
To prove that the collection $A_1,A_2,\ldots,A_q$ satisfies the requirement, take $1 \leq i \neq j \leq q$ and a non-empty set $I \subseteq [k]$. Assume, for the sake of contradiction, that the matrices $A_i$ and $A_j$ have the same sum of columns that correspond to indices in $I$. Viewing these sums as elements of $\Fset_q$, it follows that
$$\alpha_i \cdot \sum_{\ell \in I}{\beta^{\ell-1}} = \alpha_j \cdot \sum_{\ell \in I}{\beta^{\ell-1}}.$$
By linear independence, it follows that $\sum_{\ell \in I}{\beta^{\ell-1}}$ is nonzero, thus $\alpha_i = \alpha_j$, a contradiction.
\end{proof}

\begin{corollary}\label{cor:constructionB}
For every prime $p$ and a positive integer $k$, there exists a collection of $p^k$ matrices $B_1,B_2,\ldots,B_{p^k}$ in $\Fset_p^{k \times (k+1)}$ such that
\begin{enumerate}
  \item\label{item:sum_zero} for every $1 \leq i \leq p^k$, the sum of the columns of $B_i$ is the zero vector, and
  \item\label{item:subset_sum} for every $1 \leq i \neq j \leq p^k$ and every non-empty set $I \subset [k+1]$, the sum $$\sum_{\ell \in I}{B_i^{(\ell)}} + \sum_{\ell \in [k+1] \setminus I}{B_j^{(\ell)}}$$ is nonzero.
\end{enumerate}
\end{corollary}

\begin{proof}
By Lemma~\ref{lemma:constructionA}, there exists a collection of $p^k$ matrices $A_1,A_2,\ldots,A_{p^k}$ in $\Fset_p^{k \times k}$ such that for every $1 \leq i \neq j \leq p^k$ and every non-empty set $I \subseteq [k]$ it holds that $$\sum_{\ell \in I}{A_i^{(\ell)}} \neq \sum_{\ell \in I}{A_j^{(\ell)}}.$$
For every $1 \leq i \leq p^k$ define the $k$ by $k+1$ matrix $B_i$ as the matrix whose columns are
$$A_i^{(1)},A_i^{(2)},\ldots,A_i^{(k)},-\sum_{\ell=1}^{k}{A_i^{(\ell)}}.$$
The collection $B_1,B_2,\ldots,B_{p^k}$ trivially satisfies Item~\ref{item:sum_zero}. To prove that Item~\ref{item:subset_sum} is also satisfied, take $1 \leq i \neq j \leq p^k$ and a non-empty set $I \subset [k+1]$, and assume by contradiction that the sum of the columns of $B_i$ corresponding to indices in $I$ and the columns of $B_j$ corresponding to indices in $[k+1] \setminus I$ is the zero vector. Assume without loss of generality that $k+1 \notin I$, and use Item~\ref{item:sum_zero} to observe that
$$\sum_{\ell \in I}{A_i^{(\ell)}} = \sum_{\ell \in I}{B_i^{(\ell)}} = -\sum_{\ell \in [k+1] \setminus I}{B_j^{(\ell)}} = \sum_{\ell \in I}{B_j^{(\ell)}} = \sum_{\ell \in I}{A_j^{(\ell)}},$$
in contradiction to the property of the collection $A_1,A_2,\ldots,A_{p^k}$.
\end{proof}

Now, equipped with Corollary~\ref{cor:constructionB}, we are ready to show how certain collections of vectors in $\Z_D^n$ can be transformed to local PMFs over $\Fset_p$. An important property of the transformation is that it preserves optimal capacity, namely, vector collections of optimal size $(D-o(1))^n$ are transformed to local PMFs of optimal capacity $p$.
\begin{theorem}\label{thm:ConsPMF}
Let $k \geq 3$ be an integer, and let $p$ be a prime. Assume that for infinitely many values of $n$ there exists a collection $\calF$ of $(c-o(1))^n$ vectors in $\Z_{p^{k-1}}^n$ such that for every $k$ vectors $v_1,\ldots,v_k \in \calF$, not all equal, there exists an $i \in [n]$ for which
$$|\{(v_1)_i,\ldots,(v_k)_i\}| = 2.$$
Then, the local PMF capacity for $k$-cycles over $\Fset_p$ is at least $c^{1/(k-1)}$.
\end{theorem}

\begin{proof}
Let $\calF \subseteq \Z_{p^{k-1}}^n$ be a collection of $(c-o(1))^n$ vectors satisfying the condition given in the theorem. By Corollary~\ref{cor:constructionB}, applied with $k-1$, there exists a collection of $p^{k-1}$ matrices $B_1,B_2,\ldots,B_{p^{k-1}}$ in $\Fset_p^{(k-1) \times k}$ satisfying that
(1) for every $1 \leq i \leq p^{k-1}$, the sum of the columns of $B_i$ is the zero vector, and
(2) for every $1 \leq i \neq j \leq p^{k-1}$ and every non-empty set $I \subset [k]$, the sum $\sum_{\ell \in I}{B_i^{(\ell)}} + \sum_{\ell \in [k] \setminus I}{B_j^{(\ell)}}$ is nonzero.

Consider the function $f: \Z_{p^{k-1}}^n \rightarrow \Fset_p^{n(k-1) \times k}$ that maps every vector $v \in \Z_{p^{k-1}}^n$ to the concatenation of the $n$ matrices $B_{v_1},B_{v_2},\ldots,B_{v_n}$. We claim that the set $\calG = \{f(v) \mid v \in \calF\}$ is an $(n(k-1),(c-o(1))^n)$ local PMF for $k$-cycles over $\Fset_p$, thus the PMF capacity for $k$-cycles over $\Fset_p$ is at least $c^{1/(k-1)}$.

To see this, first observe that property (1) of the matrices $B_1,B_2,\ldots,B_{p^{k-1}}$ implies that the sum of the $k$ columns of every $f(v)$ is the zero vector. Second, let $f(v_1),\ldots,f(v_k)$ be $k$ elements, not all equal, of $\calG$. Our goal is to prove that the sum of the vectors $f(v_1)^{(1)},\ldots,f(v_k)^{(k)}$ is nonzero. Since $f$ is injective, the vectors $v_1,\ldots,v_k$ are not all equal, so there is an $i \in [n]$ for which
$$|\{(v_1)_i,\ldots,(v_k)_i\}|=2.$$
Hence, the $i$th blocks (of length $k-1$) of the matrices $f(v_1),\ldots,f(v_k)$ contain exactly two distinct matrices $B_j$ and $B_{j'}$. By property (2), the sum of the $i$th blocks of the vectors $f(v_1)^{(1)},\ldots,f(v_k)^{(k)}$ is nonzero, hence the sum of these vectors is also nonzero, and we are done.
\end{proof}

Theorem~\ref{thm:ConsPMF} gives us a method to derive lower bounds on the query complexity of the canonical tester for $k$-cycle-freeness over $\Fset_p$ from certain collections of vectors in $\Z^n_{p^{k-1}}$.
In the special case of $k=3$, the vectors are in $\Z_{p^2}^n$, and for every three vectors, not all equal, there is a coordinate in which the three symbols are not all equal and are not all distinct. This exactly means that the three vectors do not form a $3$-sunflower (see Section~\ref{sec:sunflowerIntro}), yielding the following result.

\begin{theorem}\label{thm:Mainp^2}
Let $p$ be a prime, and assume that for infinitely many values of $n$ there exists a collection of $(c-o(1))^n$ vectors in $\Z_{p^{2}}^n$ containing no $3$-sunflowers. Then, for every $d<\sqrt{c}$, the query complexity of the canonical tester for triangle-freeness over $\Fset_p$ for distance $\eps$ is $\Omega(1/\eps^{\alpha})$ where $\alpha = \frac{2-\log_p{d}}{1-\log_p{d}}$.
\end{theorem}

\begin{proof}
By Theorem~\ref{thm:ConsPMF}, the assumption implies that the local PMF capacity for triangles over $\Fset_p$ is at least $\sqrt{c}$. Corollary~\ref{cor:PMF2LowerBound} completes the proof.
\end{proof}

\noindent
Observe that if the Weak Sunflower Conjecture over $\Z_D$ (Conjecture~\ref{conj:weakD}) is false for $D=p^2$, it follows from the above theorem that the query complexity of the canonical tester for triangle-freeness over $\Fset_p$ for distance $\eps$ is super-polynomial in $1/\eps$, confirming Theorem~\ref{thm:FpIntro}.

For the special case of $p=3$, it is easy to get a local PMF for triangles from a collection of vectors in $\Z_3^n$ containing no $3$-sunflowers. Indeed, for such a collection $\calF$, the set $\{(x,x,x)\}_{x \in \calF}$ forms a local PMF for triangles over $\Fset_3$ of the same size. Thus, in case that the Weak Sunflower Conjecture over $\Z_3$ is false, a super-polynomial lower bound on the query complexity of the canonical tester follows. Note that this assumption can be only weaker than the assumption that the Weak Sunflower Conjecture over $\Z_9$ is false. The reason is that given a collection of $(9-o(1))^n$ vectors in $\Z_9^n$ containing no $3$-sunflowers one can replace every symbol of $\Z_9$ in the vectors by its base-$3$ representation to obtain a collection of $(3-o(1))^{2n}$ vectors in $\Z_3^{2n}$ containing no $3$-sunflowers.

We note that for $k \geq 4$ the property required in Theorem~\ref{thm:ConsPMF} from the collection $\calF \subseteq \Z_D^n$ does not coincide with freeness of $k$-sunflowers. Indeed, the collection should satisfy that for every $k$ vectors $v_1,\ldots,v_k \in \calF$, not all equal, there exists a coordinate in which they contain exactly $2$ distinct symbols. On the other hand, freeness of $k$-sunflowers means that for every such $k$ vectors, there exists a coordinate in which the number of distinct symbols is in the range from $2$ to $k-1$.

It is natural to ask if one can relate local PMFs for $k$-cycles to collections of vectors with no $k$-sunflowers for $k \geq 4$. It seems, though, that the proof technique of Theorem~\ref{thm:ConsPMF} cannot achieve this in a way that preserves optimal capacity. To see this, observe that such an extension requires a mapping from $\Z_D$ to $D$ matrices in $\Fset_p^{\ell \times k}$ for $D=p^\ell$ (because the transformation increases the length of the vectors by a factor of $\ell$, and capacity $D$ should be mapped to capacity $p$). The matrices returned by this mapping have to satisfy the following two properties: (1) the sum of the columns of each of the matrices should be zero, and (2) for every $k$ of these matrices $B_1,\ldots,B_k$, not all equal and not all distinct, the sum of the vectors $B_1^{(1)},\ldots,B_k^{(k)}$ should be nonzero. However, it is not difficult to show that such a collection of matrices does not exist for $k \geq 4$. First observe that the $p^\ell$ matrices should contain all the $p^\ell$ distinct vectors of $\Fset_p^\ell$ in each of their $k$ columns, since otherwise two of the matrices contradict property (2). Now, take two arbitrary distinct matrices $B_i$ and $B_j$, and consider the sum, say, of the first $k-2$ columns of $B_i$ and the $(k-1)$th column of $B_j$. The unique vector that completes this sum to zero is the $k$th column of one of the $p^\ell$ matrices, so we again contradict property (2).

\section{A Lower Bound on Testing $k$-Cycle-Freeness}\label{sec:Genk}

As mentioned before, the best known lower bound on the query complexity of the canonical tester for testing triangle-freeness over $\Fset_2$ for distance $\eps$ is $1/\eps^{13.239}$, as was shown by Fu and Kleinberg~\cite{FuK13}. Their proof is crucially based on a construction of uniquely solvable puzzles of Coppersmith and Winogard~\cite{CoppersmithW90}, which employs a construction of Behrend~\cite{Behrend46} of dense sets of integers with no $3$-term arithmetic progressions. In this section we generalize the lower bound of~\cite{FuK13} to testing $k$-cycle-freeness over $\Fset_p$ for every $k \geq 3$ and a prime $p$.

We need here a few notations. For a vector $v\in \Z_k^n$ we denote by $v|_{j}=\{i\in [n] \mid v_i = j\}$ the set of coordinates at which $v$ has symbol $j\in \Z_k$. Note that for every vector $v\in \Z_k^n$, the sets $v|_{1}, \ldots, v|_{k}$ form a partition of $[n]$. In case that the sets $v|_{1}, \ldots, v|_{k}$ have the same size we say that the vector $v$ is {\em balanced}. Finally, let $H$ stand for the binary entropy function, defined by $H(p) = -p\log_2{p}-(1-p)\log_2{(1-p)}$ for $0\leq p \leq 1$.

Let us start with a generalization of the construction of~\cite{CoppersmithW90}, stated below. Note that the case of $k=3$ gives the construction of~\cite{CoppersmithW90}, which implies a uniquely solvable puzzle as was defined for the purpose of fast matrix multiplication (see Definition~\ref{def:USP}).

\begin{theorem}\label{thm:GenkCW}
For every fixed integer $k \geq 3$ and a sufficiently large $n$, there exists a collection $\calF$ of $(2^{H(1/k)} - o(1))^{nk}$ balanced vectors in $\Z_k^{nk}$ such that for every $k$ vectors $v_1,\ldots,v_k \in \calF$, the sets $v_1|_1,\ldots,v_k|_k$ form a partition of $[n \cdot k]$ if and only if $v_1 = \cdots = v_k$.
\end{theorem}

\begin{remark}
The cardinality of $\calF$ in Theorem~\ref{thm:GenkCW} is optimal up to the $o(1)$ term. Indeed, the requirement on $\calF$ implies that the, say, $v|_1$'s for $v \in \calF$ are distinct subsets of size $n$ of $[n \cdot k]$, so $|\calF| \leq {nk \choose n} \approx 2^{H(1/k)nk}$.
\end{remark}

In the proof of Theorem~\ref{thm:GenkCW} we use the following extension of Behrend's result~\cite{Behrend46}.

\begin{lemma}[Lemma~3.1 in~\cite{Alon02}]\label{lemma:Behrend}
For every fixed integer $r \geq 2$ and every positive integer $m$, there exists a set $B \subseteq [m]$ of size
$$|B| \geq \frac{m}{e^{10\sqrt{\log m \log r}}}$$
with no non-trivial\footnote{A {\em trivial} solution is a solution that satisfies $x_1 = \cdots = x_{r+1}$.} solutions to the equation $x_1+x_2+\cdots+x_r = r \cdot x_{r+1}$.
\end{lemma}

\begin{proof}[ of Theorem~\ref{thm:GenkCW}]
For a sufficiently large $n$ denote $N = n \cdot k$, and let $M$ be the smallest prime which satisfies
\begin{eqnarray}\label{eq:M}
M \geq c(k) \cdot {n(k-1) \choose n,\ldots,n}^{1/(k-2)},
\end{eqnarray}
where $c = c(k)$ is a constant that depends solely on $k$ and will be determined later. As is well known, $M$ is at most twice its lower bound in~(\ref{eq:M}).
By Lemma~\ref{lemma:Behrend}, applied with $m = \lfloor M/(k-1) \rfloor$, there exists a set $B = \{b_1,\ldots,b_{|B|}\} \subseteq [m]$ of size $|B| = m^{1-o(1)} = M^{1-o(1)}$ with no non-trivial solutions to the equation
\begin{eqnarray}\label{eqn:Behrend}
x_1+x_2+\cdots+x_{k-1} = (k-1) \cdot x_{k}.
\end{eqnarray}
Since $B$ is contained in $[m]$, it contains no non-trivial solutions to Equation~(\ref{eqn:Behrend}) taken modulo $M$ as well.

Consider the set $\calI$ of all the subsets of $[N]$ of size $n$, and identify the sets in $\calI$ with their characteristic vectors in $\{0,1\}^N$. Let $w_1,\ldots,w_N$ and $c_1,\ldots,c_k$ be integers chosen at random uniformly and independently from $\Fset_M$, and denote $w=(w_1,\ldots,w_N)$. For these numbers we define $k$ mappings $\beta_1,\ldots,\beta_k : \calI \rightarrow \Fset_M$ as follows. For $1 \leq j \leq k-1$, $\beta_j$ is defined by
$$\beta_j(I) = \langle w,I \rangle +c_j~~\mbox{mod }M,$$
and $\beta_k$ is defined by
$$\beta_k(I) = \Big( \langle w,[N] \setminus I \rangle + \sum_{j=1}^{k-1}{c_j} \Big)/(k-1)~~\mbox{mod }M.$$

The construction involves two steps. First, for every $1 \leq i \leq |B|$, let $L_i$ denote the set of all $k$-tuples of sets $(I_1,\ldots,I_k) \in \calI^k$ satisfying $I_1 \cup \ldots \cup I_k = [N]$ and $\beta_j(I_j) = b_i$ for every $1 \leq j \leq k$. Second, remove from every $L_i$ all the $k$-tuples $(I_1,\ldots,I_k) \in L_i$ that share some set $I_j$ with other $k$-tuples in $L_i$, that is, satisfy $I_j = J_j$ for some $j$ and $(J_1,\ldots,J_k) \in L_i$. We denote by $L'_i \subseteq L_i$ the obtained set.

Every partition $(I_1,\ldots,I_k) \in \calI^k$ of $[N]$ can be naturally encoded by a balanced vector $v$ in $\Z_k^N$ defined by $v|_j=I_j$ for every $1 \leq j \leq k$. Define $\calF \subset \Z_k^N$ to be the set of partitions in the union $\cup_{1 \leq i \leq |B|}{L'_i}$ encoded as vectors in $\Z_k^N$. We first show that $\calF$ satisfies the property required in Theorem~\ref{thm:GenkCW}, and then analyze its expected size.

\begin{claim}
For every $k$ vectors $v_1,\ldots,v_k \in \calF$, the sets $v_1|_1,\ldots,v_k|_k$ form a partition of $[N]$ if and only if $v_1 = \cdots = v_k$.
\end{claim}

\begin{proof}
It is clear that if $v_1 = \cdots = v_k$ then $v_1|_1,\ldots,v_k|_k$ form a partition of $[N]$. For the other direction, consider $k$ vectors $v_1,\ldots,v_k \in \calF$, and denote $I_j = v_j|_j$ for every $1 \leq j \leq k$. Assume by contradiction that the vectors are not all equal and that $(I_1,\ldots,I_k) \in \calI^k$ is a partition of $[N]$. For every $1 \leq j \leq k$ denote $b_j = \beta_j(I_j)$, and observe that
$$\sum_{j=1}^{k-1}{b_j} = \sum_{j=1}^{k-1}{\beta_j(I_j)} = \sum_{j=1}^{k-1}{\langle w,I_j \rangle} +\sum_{j=1}^{k-1}{c_j} = \langle w,[N] \setminus I_k \rangle + \sum_{j=1}^{k-1}{c_j} = (k-1) \cdot \beta_k(I_k) = (k-1) \cdot b_k,$$
where all the equalities hold modulo $M$.
This implies that the numbers $b_1,\ldots,b_k \in B$ satisfy Equation~(\ref{eqn:Behrend}) modulo $M$, hence by our choice of $B$, they must be all equal. Therefore, the vectors $v_1,\ldots,v_k$ correspond to $k$ partitions in the same set $L'_i$. This implies that the partition $(I_1,\ldots,I_k)$ belongs to $L_i$ and shares a subset of $\calI$ with each of the partitions that correspond to the vectors $v_1,\ldots,v_k$. However, this implies that all these partitions were not added to $L'_i$ in the second step of the construction. Hence all the $v_i$'s are equal, in contradiction.
\end{proof}

We turn to analyze the expected size of the collection $\calF$. We start with the size of the sets $L_i$ (before performing the second step of the construction).

\begin{claim}\label{claim:sizeLi}
For every $1 \leq i \leq |B|$, the expected size of the set $L_i$ is ${nk \choose n,\ldots,n} \cdot M^{-(k-1)}$.
\end{claim}

\begin{proof}
Fix $1 \leq i \leq |B|$, and let $(I_1,\ldots,I_k) \in \calI^k$ be a partition of $[N]$. Recall that $(I_1,\ldots,I_k)$ is added to $L_i$ if $\beta_j(I_j) = b_i$ for every $1 \leq j \leq k$. We claim that this happens with probability $M^{-(k-1)}$. Indeed, the $k-1$ events $\beta_j(I_j) = b_i$, $1 \leq j \leq k-1$, are independent, each of them occurs with probability $M^{-1}$, and once they all occur, it follows that $\beta_k(I_k) = b_i$ as well. The number of partitions of $[N]$ in $\calI^k$ is ${nk \choose n,\ldots,n}$, so by linearity of expectation the claim follows.
\end{proof}

We now turn to estimate the expected number of partitions that are removed from every $L_i$ in the second step of the construction. To do so, we have to consider the probability of two distinct partitions in $\calI^k$ that share some subset to belong to $L_i$. However, this probability depends on the specific pair of partitions. Indeed, sharing more than one subset of the partitions, or even a certain union of the subsets, might increase this probability. Hence, we need the following definition.

\begin{definition}\label{def:similar}
Let $t_1 \leq \ldots \leq t_\ell$ be $\ell$ positive integers satisfying $\sum_{r=1}^{\ell}{t_r}=k$. We say that two partitions $(I_1,\ldots,I_k)$ and $(J_1,\ldots,J_k)$ of $[N]$ in $\calI^k$ are {\em $(t_1,\ldots,t_\ell)$-similar} if there exists a partition of $[k]$ into $\ell$ sets $T_1,\ldots,T_\ell$ of sizes $t_1,\ldots,t_\ell$ respectively, such that for some permutation $\pi:[k] \rightarrow [k]$,
\begin{eqnarray}\label{eq:similar}
\cup_{i \in T_r}{I_i} = \cup_{i \in T_r}{J_{\pi(i)}}
\end{eqnarray}
for every $1 \leq r \leq \ell$, and, in addition, no refinement of the partition $T_1,\ldots,T_\ell$ satisfies~(\ref{eq:similar}) for any permutation $\pi$.
\end{definition}

\begin{claim}\label{claim:probability}
Let $(I_1,\ldots,I_k)$ and $(J_1,\ldots,J_k)$ be distinct $(t_1,\ldots,t_\ell)$-similar partitions of $[N]$ in $\calI^k$ for some $\ell$ positive integers $t_1 \leq \ldots \leq t_\ell$ satisfying $\sum_{r=1}^{\ell}{t_r}=k$. Then, for every $1 \leq i \leq |B|$, the following holds.
\begin{enumerate}
  \item\label{itm:l<k} For $1 \leq \ell \leq k-1$, the probability that the two partitions are in $L_i$ is at most $M^{-(k-1)} \cdot M^{-(k-\ell)}$.
  \item\label{itm:l=k} For $\ell = k$, the probability that the two partitions are in $L_i$ is at most $M^{-(k-1)} \cdot M^{-1}$.
\end{enumerate}
\end{claim}

\begin{proof}
Let $(I_1,\ldots,I_k)$ and $(J_1,\ldots,J_k)$ be distinct $(t_1,\ldots,t_\ell)$-similar partitions of $[N]$ in $\calI^k$, and let $T_1,\ldots,T_\ell$ and $\pi$ be the corresponding partition and permutation of $[k]$ as in Definition~\ref{def:similar}.

For Item~\ref{itm:l<k}, we fix the values of $c_1,\ldots,c_k$ and analyze the probability that the two partitions $(I_1,\ldots,I_k)$ and $(J_1,\ldots,J_k)$ are in $L_i$ over the random choice of $w_1,\ldots,w_N$. First, notice that the $k-1$ events $\beta_j(I_j) = b_i$ for $1 \leq j \leq k-1$ are independent, occur with probability $M^{-1}$ each, and imply the event $\beta_k(I_k) = b_i$. So the probability that $(I_1,\ldots,I_k) \in L_i$ is $M^{-(k-1)}$.
For $1 \leq r \leq \ell$, let $A_r$ denote the event that the equalities $\beta_{\pi(j)}(J_{\pi(j)}) = b_i$ hold for every $j \in T_r$. It can be shown that for every $1 \leq r \leq \ell$, the probability that $A_r$ occurs conditioned on the event $(I_1,\ldots,I_k) \in L_i$ and on $A_1,\ldots,A_{r-1}$ is $M^{-(t_r-1)}$. Indeed, since no refinement of the partition $T_1,\ldots,T_\ell$ satisfies the condition of Definition~\ref{def:similar}, it follows that $t_r-1$ of the equalities of $A_r$ are independent, occur with probability $M^{-1}$ each, and might imply the last one. This can be verified by observing that every vector $J_{j}$ that corresponds to such an equality is linearly independent of the vectors that correspond to the previously considered equalities. Hence, the probability that the two $k$-tuples are in $L_i$ is at most
$$M^{-(k-1)} \cdot M^{-\sum_{r=1}^{\ell}{(t_r-1)}} = M^{-(k-1)} \cdot M^{-(k-\ell)}.$$

For Item~\ref{itm:l=k}, take $\ell = k$ and notice that in this case, $(J_1,\ldots,J_k)$ is a permutation of $(I_1,\ldots,I_k)$, so we have $t_1 = \cdots = t_\ell = 1$. The probability that $(I_1,\ldots,I_k) \in L_i$ is again $M^{-(k-1)}$. Since the two partitions are distinct, there are distinct $j,j'$ for which $I_j = J_{j'}$. Assume, without loss of generality, that $j' < k$. By the randomness of the choice of $c_{j'}$, the probability that $\beta_{j'}(J_{j'}) = b_i$, conditioned on $(I_1,\ldots,I_k) \in L_i$, is $M^{-1}$. Therefore, the probability that both the partitions are in $L_i$ is at most $M^{-(k-1)} \cdot M^{-1}$.
\end{proof}

\begin{claim}\label{claim:L'}
For every $1 \leq i \leq |B|$, the expected size of the set $L'_i$ is at least $\frac{1}{2} \cdot {nk \choose n,\ldots,n} \cdot M^{-(k-1)}$.
\end{claim}

\begin{proof}
By Claim~\ref{claim:sizeLi}, the expected size of $L_i$ is ${nk \choose n,\ldots,n} \cdot M^{-(k-1)}$. We turn to bound from above the expected number of partitions removed from $L_i$ in the second step. Notice that if two partitions of $[N]$ in $\calI^k$ share a subset then they are $(t_1,\ldots,t_\ell)$-similar for some $\ell \geq 2$ positive integers $t_1 \leq \ldots \leq t_\ell$ satisfying $\sum_{r=1}^{\ell}{t_r} =k$ (in fact, we also know that at least one of the $t_r$'s equals $1$). Therefore, we turn to bound the expected number of (ordered) pairs of $(t_1,\ldots,t_\ell)$-similar partitions of $[N]$ in $\calI^k$ for some $t_1,\ldots,t_\ell$ as above.

We start with the case $2 \leq \ell \leq k-1$. Fix a partition $(I_1,\ldots,I_k)$ of $[N]$ in $\calI^k$, and $\ell$ positive integers $t_1 \leq \ldots \leq t_\ell$ satisfying $\sum_{r=1}^{\ell}{t_r} =k$. The number of partitions $(J_1,\ldots,J_k)$ of $[N]$ in $\calI^k$ which are $(t_1,\ldots,t_\ell)$-similar to $(I_1,\ldots,I_k)$, associated with certain partition $T_1,\ldots,T_\ell$ and permutation $\pi$ of $[k]$, is at most
$${n \cdot t_1 \choose n,\ldots,n} \cdot \ldots \cdot {n \cdot t_\ell \choose n,\ldots,n} \leq {n(k - \ell +1) \choose n,\ldots,n}.$$
By Item~\ref{itm:l<k} of Claim~\ref{claim:probability}, the probability that two such partitions are in $L_i$ is at most $M^{-(k-1)} \cdot M^{-(k-\ell)}$. It follows that the expected total number of pairs of $(t_1,\ldots,t_\ell)$-similar partitions in $L_i$ for {\em some} $t_1 \leq \ldots \leq t_\ell$ as above ($2 \leq \ell \leq k-1$) is at most
$$k^{O(k)} \cdot {nk \choose n,\ldots,n} \cdot {n(k - \ell+1) \choose n,\ldots,n} \cdot M^{-(k-1)} \cdot M^{-(k-\ell)},$$
where the $k^{O(k)}$ term counts all the possible choices of numbers $t_1,\ldots,t_\ell$, partitions $T_1,\ldots,T_\ell$, and permutations $\pi$ of $[k]$.
Now, consider the case $k=\ell$, that is, $t_1 = \cdots = t_\ell = 1$. Using Item~\ref{itm:l=k} of Claim~\ref{claim:probability}, the expected number of $(1,\ldots,1)$-similar partitions in $L_i$ is at most
$$k^{O(k)} \cdot {nk \choose n,\ldots,n} \cdot M^{-(k-1)}\cdot M^{-1}.$$
Therefore, by linearity of expectation, the expected number of partitions removed from $L_i$ is at most
$${nk \choose n,\ldots,n} \cdot M^{-(k-1)} \cdot \Big(k^{O(k)} \cdot M^{-1} + \sum_{\ell=2}^{k-1}{k^{O(k)} \cdot {n(k-\ell+1) \choose n,\ldots,n} \cdot M^{-(k-\ell)}} \Big).$$

Choosing the constant $c(k)$ in~(\ref{eq:M}) to be sufficiently large, we have $k^{O(k)} \cdot M^{-1} \leq \frac{1}{2k}$.
We turn to prove that for every $2 \leq \ell \leq k-1$,
\begin{eqnarray}\label{eqn:termCW}
{k^{O(k)} \cdot {n(k-\ell+1) \choose n,\ldots,n} \cdot M^{-(k-\ell)}} \leq \frac{1}{2k},
\end{eqnarray}
as this implies, combined with Claim~\ref{claim:sizeLi}, that the expected size of $L'_i$ is at least
$${nk \choose n,\ldots,n} \cdot M^{-(k-1)}- \frac{1}{2}\cdot {nk \choose n,\ldots,n} \cdot M^{-(k-1)} = \frac{1}{2}\cdot {nk \choose n,\ldots,n} \cdot M^{-(k-1)}.$$
For~(\ref{eqn:termCW}), observe that our choice of $M$ satisfies
$${k^{O(k)} \cdot {n(k-\ell+1) \choose n,\ldots,n} \cdot M^{-(k-\ell)}} \leq \frac{1}{2k} \cdot {n(k-\ell+1) \choose n,\ldots,n} \cdot {n(k-1) \choose n,\ldots,n}^{-\frac{k-\ell}{k-2}} \leq \frac{1}{2k},$$
where the first inequality holds for a sufficiently large constant $c(k)$ in~(\ref{eq:M}), and the second follows from the inequality ${n(k-\ell+1) \choose n,\ldots,n}^{1/(k-\ell)} \leq {n(k-1) \choose n,\ldots,n}^{1/(k-2)}$ that holds for every $\ell \geq 2$ by monotonicity of the geometric mean.
\end{proof}

Finally, using Claim~\ref{claim:L'}, we conclude that there exists a choice of $w_1,\ldots,w_N$ and $c_1,\ldots,c_k$ for which
$$|\calF| = \sum_{i=1}^{|B|}{|L'_i|} \geq \frac{1}{2} \cdot {nk \choose n,\ldots,n} \cdot M^{-(k-1)} \cdot |B| \geq  {nk \choose n,\ldots,n} \cdot M^{-(k-2)-o(1)} \geq {nk \choose n}^{1-o(1)}.$$
By standard estimations of Binomial coefficients, this completes the proof of Theorem~\ref{thm:GenkCW}.
\end{proof}

\begin{corollary}
For every $k \geq 3$ and a prime $p$, the local PMF capacity for $k$-cycles over $\Fset_p$ is at least $2^{H(1/k)}$.
\end{corollary}

\begin{proof}
By Theorem~\ref{thm:GenkCW}, for every sufficiently large $n$, there exists a collection $\calF$ of $(2^{H(1/k)} - o(1))^{nk}$ balanced vectors in $\Z_k^{nk}$ such that for every $k$ vectors $v_1,\ldots,v_k \in \calF$, the sets $v_1|_1,\ldots,v_k|_k$ form a partition of $[n \cdot k]$ if and only if $v_1 = \cdots = v_k$. For every vector $v \in \calF$ consider the $k$-tuple of vectors, whose first $k-1$ vectors are the characteristic vectors of $v|_1,v|_2,\ldots,v|_{k-1}$, and the last one is the characteristic vector of $[n] \setminus v|_k$ multiplied by $-1$ (modulo $p$). Observe that the collection of all $k$-tuples obtained in this way from the vectors of $\calF$ is an $(nk,|\calF|)$ local PMF for $k$-cycles over $\Fset_p$. Hence, the local PMF capacity for $k$-cycles over $\Fset_p$ is at least $2^{H(1/k)}$, as required.
\end{proof}
\noindent
The above corollary, combined with Lemma~\ref{lemma:LowerBound(m,n)}, completes the proof of Theorem~\ref{thm:GenkIntro}.
\bibliographystyle{abbrv}
\bibliography{sunflower}

\end{document}